\documentclass[journal,onecolumn]{IEEEtran}

\usepackage{amsmath, epsfig, citesort}
\usepackage{amssymb}
\usepackage{amsfonts}
\usepackage{latexsym}
\usepackage{longtable}
\usepackage{tabularx}
\usepackage{pgf}
\usepackage{tikz}
\usepackage{dblfloatfix}
\usetikzlibrary{arrows, automata, positioning, calc, shapes}

\newtheorem{theorem}{Theorem}

\newtheorem{remark}[theorem]{Remark}

\newtheorem{lemma}[theorem]{Lemma}
\newtheorem{conjecture}[theorem]{Conjecture}

\newtheorem{example}[theorem]{\it Example}

\begin{document}

\title{On private information retrieval array codes}

\author{Yiwei Zhang, Xin Wang, Hengjia Wei and Gennian Ge
\thanks{The research of G. Ge was supported by the National Natural Science Foundation of China under Grant Nos. 11431003 and 61571310.}
\thanks{Y. Zhang is with the School of Mathematical Sciences, Capital Normal University, Beijing 100048, China (email: rexzyw@163.com).}
\thanks{X. Wang is with the School of Mathematical Sciences, Zhejiang University, Hangzhou 310027, China (email: 11235062@zju.edu.cn).}
\thanks{H. Wei is with the School of Physical and Mathematical Sciences, Nanyang Technological University, Singapore (email: ven0505@163.com).}
\thanks{G. Ge is with the School of Mathematical Sciences, Capital Normal University,
Beijing 100048, China. He is also with Beijing Center for Mathematics and Information Interdisciplinary Sciences, Beijing 100048, China (e-mail: gnge@zju.edu.cn).}
}

\date{}\maketitle

\begin{abstract}
Given a database, the private information retrieval (PIR) protocol allows a user to make queries to several servers and retrieve a certain item of the database via the feedbacks, without revealing the privacy of the specific item to any single server. Classical models of PIR protocols require that each server stores a whole copy of the database. Recently new PIR models are proposed with coding techniques arising from distributed storage system. In these new models each server only stores a fraction $1/s$ of the whole database, where $s>1$ is a given rational number. PIR array codes are recently proposed by Fazeli, Vardy and Yaakobi to characterize the new models. Consider a PIR array code with $m$ servers and the $k$-PIR property (which indicates that these $m$ servers may emulate any efficient $k$-PIR protocol). The central problem is to design PIR array codes with optimal rate $k/m$. Our contribution to this problem is three-fold. First, for the case $1<s\le 2$, although PIR array codes with optimal rate have been constructed recently by Blackburn and Etzion, the number of servers in their construction is impractically large. We determine the minimum number of servers admitting the existence of a PIR array code with optimal rate for a certain range of parameters. Second, for the case $s>2$, we derive a new upper bound on the rate of a PIR array code. Finally, for the case $s>2$, we analyze a new construction by Blackburn and Etzion and show that its rate is better than all the other existing constructions.
\end{abstract}


\section{Introduction}

The private information retrieval (PIR) protocol is first introduced in \cite{chor}. The classical model is as follows. Suppose we have an $n$-bit database and a set of $k$ servers, each storing a whole copy of the database, so the total {\it storage overhead} is $nk$. A $k$-server PIR protocol will allow a user to retrieve a data item while each server (as long as they do not collude) has no information about which item is retrieved. For example, suppose the database is $\mathbf{x}=(x_1,x_2,\dots,x_n)$ and a user wants to retrieve $x_i$. In a $2$-server PIR protocol, the user may randomly pick a vector $\mathbf{v}\in \{0,1\}^n$. The first server receives the query $\mathbf{v}$ and responds to the user with $\mathbf{v}\cdot\mathbf{x}$. The second server receives the query $\mathbf{v}+\mathbf{e_i}$ and responds with $(\mathbf{v+e_i})\cdot\mathbf{x}$. Then the user may retrieve $x_i=(\mathbf{v+e_i})\cdot\mathbf{x}-\mathbf{v}\cdot\mathbf{x}$. Each server itself does not know which item is retrieved since $\mathbf{v}$ is a random vector.

Recently, PIR protocols have been combined with techniques and ideas arising from distributed storage system \cite{augot,chan,fanti,shah,tajeddine}. Instead of storing a complete copy of the database in each server, in the newly proposed models each server only stores a fraction of the database. A breakthrough by Fazeli, Vardy and Yaakobi \cite{fazeli,fazeli2} shows that $m$ servers (for some $m>k$) may emulate a $k$-server protocol with storage overhead significantly smaller than $nk$. Continuing the example above, let three servers store the following fractions of database respectively: $\mathbf{x'}=(x_1,\dots,x_{n/2})$, $\mathbf{x''}=(x_{n/2+1},\dots,x_{n})$ and $\mathbf{x'}+\mathbf{x''}$. Without loss of generality assume that a user wants to retrieve $x_i$ with $1\le i \le n/2$. The user may randomly pick a vector $\mathbf{u}\in \{0,1\}^{n/2}$ and the queries for the three servers are correspondingly $\mathbf{u}$, $\mathbf{u}+\mathbf{e_i}$ and $\mathbf{u}+\mathbf{e_i}$. Then by calculating $x_i=-\mathbf{u}\cdot\mathbf{x'}-(\mathbf{u}+\mathbf{e_i})\cdot\mathbf{x''}+(\mathbf{u}+\mathbf{e_i})\cdot(\mathbf{x'}+\mathbf{x''})$, the user successfully retrieves the item $x_i$ without revealing its privacy. Compared with the original model, the storage overhead reduces from $2n$ to $\frac{3n}{2}$.

In \cite{fazeli} the problem of designing PIR protocols is reformulated as designing a corresponding PIR array code, which is defined as follows. Given positive integers $t$, $m$, $p$ and $k$, a {\it $[t\times m,p]$ array code} is a $t\times m$ array, where each entry is a linear combination of $\{x_1,\dots,x_p\}$ (we may view each $x_i$ as an element in a certain finite field $\mathbb{F}$). The array code has the {\it $k$-PIR property} if for every $i\in \{1,2,\dots,p\}$ there exist $k$ pairwise disjoint subsets $S_1,S_2,\dots,S_k$ of columns such that the entries in each $S_j$ could linearly span $x_i$, $1\le j \le k$. We further call such an array code a $[t\times m,p]$ $k$-PIR array code. For example, the following is a $[3\times6,6]$ $4$-PIR array code:

\begin{table}[!hb]
\centering
\begin{tabular}{|c|c|c|c|c|c|}
  \hline
  $x_1$ & $x_2$ & $x_3$ & $x_4$ & $x_5$ & $x_6$ \\\hline
  $x_2$ & $x_3$ & $x_4$ & $x_5$ & $x_6$ & $x_1$ \\\hline
  $x_3+x_4$ & $x_4+x_5$ & $x_5+x_6$ & $x_6+x_1$ & $x_1+x_2$ & $x_2+x_3$ \\
  \hline
\end{tabular}
\end{table}

We may verify the $4$-PIR property directly. For example, $x_1$ may be spanned by $S_1=\{1\}$, $S_2=\{6\}$, $S_3=\{2,5\}$ (by $(x_1+x_2)-x_2$) and $S_4=\{3,4\}$ (by $(x_1+x_6)+x_5-(x_5+x_6)$).

The relation of a $[t\times m,p]$ $k$-PIR array code with a PIR protocol is as follows. The $n$-bit database is partitioned into $p$ parts $\{x_1,x_2,\dots,x_p\}$, each part encoded as an element of a certain finite field $\mathbb{F}$. A column of the array corresponds to a server. Each server has $t$ {\it cells} storing the linear combinations of $\{x_1,x_2,\dots,x_p\}$ suggested by the entries. In \cite{fazeli} it is shown that the $k$-PIR property allows the servers to emulate all known efficient $k$-server PIR protocols. The storage overhead in this scheme is then $ntm/p$, better than $nk$ if the array code is good enough (namely $tm/p<k$). Let $s$ be the ratio between the size of the whole database and that of the data stored on each server, i.e., $s=\frac{n}{nt/p}=p/t$. The goal is to minimize the storage overhead, so we would like that $\frac{nk}{ntm/p}=s\frac{k}{m}$ is as large as possible. The PIR rate of such an array code is then defined as $k/m$.

Note that given a PIR array code, each server could actually span a subspace $V$ of $\mathbb{F}^p$ of dimension at most $t$ using the information in its $t$ cells. Changing the cells to produce a new spanning set for $V$, or even to replace $V$ by a larger subspace containing $V$, will not harm the $k$-PIR property. So without loss of generality we shall follow two assumptions posed in \cite{blackburn}:

$\bullet$ if $x_i$ can be derived by a single server alone, then $x_i$ is stored as the value of one of the cells of the server;

$\bullet$ the data stored in any server's cells are linearly independent, i.e., the subspace spanned by the $t$ cells has dimension $t$.

A further reasonable assumption is to make the PIR array code as simple as possible. We assume that if $x_i$ can be derived by a single server alone, then except for the singleton cell $x_i$, the item $x_i$ does not appear in any other cell of the server.

Now the general problem is as follows. Given $s$ and $t$ (so $p=st$ is also given), we want to build a $[t\times m,p]$ $k$-PIR array code with the largest rate $k/m$, denoted as $g(s,t)$. Further we want to analyze $g(s)=\overline{\lim}_{t\rightarrow\infty}g(s,t)$. Below we will list several results regarding this problem, the first two of which can be derived from \cite{fazeli}.

\begin{theorem}
For any given positive integer $s$, $g(s,1)\le \frac{2^{s-1}}{2^s-1}$, with equality if and only if $k$ is divisible by $2^{s-1}$.
\end{theorem}

\begin{theorem}
For any integer $s\ge 3$, we have $g(s,s-1)\ge \frac{s}{2s-1}$.
\end{theorem}

Recently, this problem receives the attention from Blackburn and Etzion. They aim to construct PIR array codes with optimal PIR rate. Some of their main results in \cite{blackburn} are listed below, including: two upper bounds regarding $g(s,t)$ and $g(s)$; constructions meeting the upper bound for $1<s\le 2$; and several constructions for the case $s>2$.

\begin{theorem}
For each rational number $s>1$ we have that $g(s)\le\frac{s+1}{2s}$. There is no $t$ such that $g(s,t)=\frac{s+1}{2s}$.
\end{theorem}

\begin{theorem} \label{upperbound}
For any integer $t\ge2$ and any positive integer $d$, with $s=1+\frac{d}{t}$ and $p=t+d$, we have
$$g(1+\frac{d}{t},t)\le\frac{(2d+1)t+d^2}{(t+d)(2d+1)}=1-\frac{d^2+d}{(t+d)(2d+1)}.$$
Moreover, when $1<s\le2$, this upper bound is tight. That is, $g(2,t)=\frac{3t+1}{4t+2}$ and for $t\ge 2$, $1\le d \le t-1$, $g(1+\frac{d}{t},t)=\frac{(2d+1)t+d^2}{(t+d)(2d+1)}$.
\end{theorem}

\begin{theorem} \label{severalconstr}
There exist PIR array codes satisfying the following parameters:
\begin{enumerate}
    \item Let $s=\frac{rt-(r-2)r-1}{t}$ and then $p=rt-(r-2)r-1$, where $3\le r \le t$. Then $g(s,t)\ge \frac{1}{2}+\frac{t-r+1}{2(rt-(r-2)r-1)}$.
    \item Let $s=r+d/t$ and then $p=rt+d$, where $r\ge 2$ is an integer, $t\ge r$, $1\le d \le t-1$. Then $g(s,t)\ge 1-\frac{(rt+d-t+r)(rt+d-t)}{(rt+d)(2rt+2d-2t+r)}$.
    \item Let $s>2$ be an integer and $t\ge s$. Then $g(s,t)\ge \frac{st+t+1}{s(2t+1)}$.
    \item Let $s>2$ be an integer. Let $(s-1)t=lb$ and $t\ge l+b$, where $l$ and $b$ are positive integers. Then $g(s,t)\ge \frac{s+1}{2s}-\frac{l}{2st}$.
\end{enumerate}
\end{theorem}

Our contribution to PIR array codes is three-fold.

First, given $t\ge2$, $1\le d\le t$ and $s=1+\frac{d}{t}$ (then $1<s\le2$), although PIR array codes meeting the upper bound have been constructed by Blackburn and Etzion, the number of columns in their array codes is impractically large ($m={{t+d}\choose t}\frac{v}{d}+{{t+d}\choose d+1}\frac{v}{t}$ where $v$ is the least common multiple of $d$ and $t$). That is, a corresponding scheme requires a lot of servers in order to meet the optimal rate. A scheme with a small number of servers is of interest due to applicable reasons. 
Therefore we shall consider the following problem: what is the smallest number of servers $m$ such that an array code with optimal rate $k/m$ exists? In \cite{blackburn} the case $d=1$ is solved. In this paper, we show that when $t>d^2-d$, the smallest number of servers such that an array code with optimal rate exists is $m=p(2d+1)/\omega$ where $\omega=\gcd\{d^2+d,p(2d+1)\}$.

Second, for the case $s>2$, we derive a new upper bound on the rate of a PIR array code which improves the result shown in Theorem \ref{upperbound}.

Finally, for the case $s>2$, a new construction appears in \cite[Section 4]{blackburn} by Blackburn and Etzion. We deeply analyze their construction in the following aspects. First, a minor problem of their construction is that the number of servers is very large. By a slight modification, we propose another construction which has a much smaller number of servers compared to the original construction, with only a slight sacrifice in the rate. Second, we shall demonstrate that both constructions can produce codes of larger rate than all the other existing ones in Theorem \ref{severalconstr}. Finally, we have some discussions regarding the potential optimality of this construction. 

The rest of the paper is organized as follows. In Section \ref{minservers} we analyze the case $1<s\le2$ and determine the minimum number of servers needed to implement an array code with optimal rate for $t>d^2-d$. In Section \ref{newupperbound} we derive a new upper bound on the rate for $s>2$. In Section \ref{constru}, for the case $s>2$, we analyze the construction by Blackburn and Etzion in several aspects. Section \ref{conclusion} concludes the paper.

\section{$1<s\le2$: optimal PIR array codes with minimum number of servers} \label{minservers}

In this section we deal with the case $1<s\le2$, where $s=1+\frac{d}{t}$. In this case PIR array codes with optimal rate have been constructed by Blackburn and Etzion \cite{blackburn}. However the number of servers $m$ in their constructions is impractically large ($m={{t+d}\choose t}\frac{v}{d}+{{t+d}\choose d+1}\frac{v}{t}$ where $v$ is the least common multiple of $d$ and $t$). For applicable reasons, array codes with a smaller number of servers are of interest. 
A natural question is to construct PIR array codes with minimum number of servers while maintaining the optimal rate.

Let $s=1+d/t$, where $1\le d \le t$, then $1<s\le2$ and $p=ts=t+d$. The upper bound of the rate has been shown to be $1-\frac{d^2+d}{p(2d+1)}$ in Theorem \ref{upperbound}. Let $\omega$ be the greatest common divisor of $d^2+d$ and $p(2d+1)$, then the smallest possible number of servers for a PIR array code with optimal rate is then $p(2d+1)/\omega$. In this section we are going to prove for a certain range of parameters that there do exist such PIR array codes with $p(2d+1)/\omega$ servers.

Since $\omega|d^2+d$, then we can split $\omega$ as $\omega=\omega_1\omega_2$, where $\omega_1|d$, $\omega_2|(d+1)$. Moreover, since $d$ and $d+1$ are relatively prime, then $\omega_1$ and $\omega_2$ are relatively prime. Denote $d=\omega_1d_1$ and $d+1=\omega_2d_2$. Furthermore, since $\gcd(d,2d+1)=1$ and $\gcd(d+1,2d+1)=1$, then we can deduce that $\omega|p$. Denote $p=\mu\omega=\mu\omega_1\omega_2$ and then the desired number of servers will be $m=\frac{p(2d+1)}{\omega}=\mu(2d+1)$.

We first define two types of servers. For a server of the first type, every cell of the server contains a singleton item in $\{x_1,x_2,\dots,x_p\}$ and we call it a singleton server. Such a server contains $t$ singleton cells, say $\{y_1,y_2,\dots,y_t\}$, and we denote this server by $\overline{A}$, where $A=\{x_1,x_2,\dots,x_p\}\backslash\{y_1,y_2,\dots,y_t\}$.  For a server of the second type, $t-1$ cells of the server contain a singleton item, say $\{z_1,z_2,\dots,z_{t-1}\}$, and the remaining cell contains the summation of all the items except for $\{z_1,z_2,\dots,z_{t-1}\}$. We call it a $\Sigma$-server and denote it by $\Sigma B$, where $B=\{x_1,x_2,\dots,x_p\}\backslash\{z_1,z_2,\dots,z_{t-1}\}$. The PIR array code we shall construct consists of these two types of servers defined above. Within this section all indices are reduced modulo $p$.\\

\noindent\fbox{%
  \parbox{\textwidth}{%
Construction (given $t$, $d$ satisfying $t>d^2-d$):\\

1. We have $\mu(d+1)$ singleton servers as follows. For $0\le j \le \mu\omega_2-1$, define $A_j=\{ x_{j+\alpha+\beta\mu\omega_2}: 0\le \alpha\le d_1-1, 0\le \beta\le \omega_1-1 \}$. Since $d_1\omega_1=d<p=\mu \omega_1\omega_2$, we have $d_1<\mu \omega_2$ and thus there are no repeated items in each $A_j$. Therefore the cardinality of $A_j$ is exactly $d_1\omega_1=d$. The $\mu(d+1)=\mu\omega_2d_2$ singleton servers are the servers $\overline{A_0},\overline{A_1},\dots,\overline{A_{\mu\omega_2-1}}$, each appearing $d_2$ times.\\

2. We have $\mu d$ $\Sigma$-servers as follows. For $0\le j \le \mu\omega_1-1$, define $B_j=\{ x_{j+\gamma d_1+\lambda \mu\omega_1}: 0\le \gamma\le d_2-1, 0\le \lambda\le\omega_2-1 \}$. Since $t>d^2-d$, we have $d_1\omega_1(d_2-1)\omega_2\leq d^2< p=\mu \omega_1\omega_2$, so $d_1(d_2-1)< \mu $ and thus there are no repeated items in each $B_j$. Therefore the cardinality of $B_j$ is exactly $d_2\omega_2=d+1$. The $\mu d=\mu\omega_1d_1$ $\Sigma$-servers are the servers $\Sigma B_0,\Sigma B_1,\dots,\Sigma B_{\mu\omega_1-1}$, each appearing $d_1$ times.
  }%
}

\medskip

Next we shall show that the construction above produces PIR array codes with optimal rate, for $t>d^2-d$. We discuss in two separated cases, $t\ge d^2$ and $d^2-d<t<d^2$.

\subsection{$t\ge d^2$}

\begin{theorem} \label{first}
For $t\ge d^2$, there exist $k$-PIR array codes with $m=\mu(2d+1)$ servers such that the rate $\frac{k}{m}$ equals $g(s,t)=1-\frac{d^2+d}{p(2d+1)}$.
\end{theorem}

\begin{proof}
We claim that when $t\ge d^2$, for any $A_{j_1}$ and $B_{j_2}$, $|A_{j_1}\bigcap B_{j_2}|\le 1$. Suppose otherwise, we have at least two distinct items in $A_{j_1}\bigcap B_{j_2}$, that is, there exist $\alpha_1,\alpha_2,\beta_1,\beta_2,\gamma_1,\gamma_2,\lambda_1,\lambda_2$ such that $$j_1+\alpha_1+\beta_1\mu\omega_2\equiv j_2+\gamma_1d_1+\lambda_1\mu\omega_1 \pmod{p}$$ and $$j_1+\alpha_2+\beta_2\mu\omega_2\equiv j_2+\gamma_2d_1+\lambda_2\mu\omega_1 \pmod{p}.$$
Do subtractions using these two equations above, we can get $$(\alpha_2-\alpha_1)+(\beta_2-\beta_1)\mu\omega_2\equiv (\gamma_2-\gamma_1)d_1+(\lambda_2-\lambda_1)\mu\omega_1 \pmod{p}.$$
Then we have $$(\alpha_2-\alpha_1)\equiv (\gamma_2-\gamma_1)d_1 \pmod{\mu}.$$
When $t\ge d^2$, $\mu=\frac{p}{\omega_1\omega_2}\ge \frac{d(d+1)}{\omega_1\omega_2}=d_1d_2$. Since $1-d_1\le \alpha_2-\alpha_1 \le d_1-1$ and $1-d_2 \le \gamma_2-\gamma_1 \le d_2-1$, so $(\alpha_2-\alpha_1)\equiv (\gamma_2-\gamma_1)d_1 \pmod{\mu}$ holds if and only if $\alpha_2-\alpha_1=\gamma_2-\gamma_1=0$. Then we have $(\beta_2-\beta_1)\mu\omega_2\equiv (\lambda_2-\lambda_1)\mu\omega_1 \pmod{p}$ and equivalently $$(\beta_2-\beta_1)\omega_2\equiv (\lambda_2-\lambda_1)\omega_1 \pmod{\omega_1\omega_2}.$$
Since $\omega_1$ and $\omega_2$ are relatively prime, $1-\omega_1\le\beta_2-\beta_1\le\omega_1-1$ and $1-\omega_2\le \lambda_2-\lambda_1\le \omega_2-1$, then $(\beta_2-\beta_1)\omega_2\equiv (\lambda_2-\lambda_1)\omega_1 \pmod{\omega_1\omega_2}$ holds if and only if $\beta_2-\beta_1=\lambda_2-\lambda_1=0$. Now we arrive at a contradiction to the existence of two distinct items in $A_{j_1}\bigcap B_{j_2}$. Therefore $|A_{j_1}\bigcap B_{j_2}|\le 1$ as claimed.

To analyze the $k$-PIR property, it suffices to analyze how to span an item $x_0$ since obviously the construction is symmetric for all items $\{x_0,x_1,\dots,x_p\}$. In the $\mu(d+1)$ singleton servers there are totally $t\mu(d+1)$ singleton cells and $\frac{t\mu(d+1)}{p}$ among them contain the singleton $x_0$. So $\frac{t\mu(d+1)}{p}=\frac{td_2}{\omega_1}$ singleton servers contain a singleton $x_0$ and the rest $\frac{d\mu(d+1)}{p}=d_1d_2$ singleton servers do not. In the $\mu d$ $\Sigma$-servers there are totally $(t-1)\mu d$ singleton cells and $\frac{(t-1)\mu d}{p}$ among them contain the singleton $x_0$. So $\frac{(t-1)\mu d}{p}=\frac{(t-1)d_1}{\omega_2}$ $\Sigma$-servers contain a singleton $x_0$ and the rest $\frac{(d+1)\mu d}{p}=d_1d_2$ $\Sigma$-servers do not. Arbitrarily choose one of the $d_1d_2$ singleton servers not containing the singleton $x_0$, then the server should be $\overline{A_{j_1}}$ where $x_0\in A_{j_1}$. Arbitrarily choose one of the $d_1d_2$ $\Sigma$-servers not containing the singleton $x_0$, then the server should be $\Sigma B_{j_2}$ where $x_0\in B_{j_2}$. Since we have claimed $|A_{j_1}\bigcap B_{j_2}|\le 1$, then the server $\overline{A_{j_1}}$ knows all the items except for $x_0$ among the summation stored in the non-singleton cell of the server $\Sigma B_{j_2}$, so they two together can span the item $x_0$.

So we can finally deduce that $k=\frac{td_2}{\omega_1}+\frac{(t-1)d_1}{\omega_2}+d_1d_2=\frac{d^2+2td+t}{\omega_1\omega_2}$. Thus the rate of this array code is $k/m=\frac{d^2+2td+t}{\mu(2d+1)\omega_1\omega_2}=1-\frac{d^2+d}{p(2d+1)}$, meeting the upper bound.
\end{proof}

\begin{remark}
Build a bipartite graph where the first part of vertices corresponds to the set of singleton servers not containing the singleton $x_0$ and the second part of vertices corresponds to the set of $\Sigma$-servers not containing the singleton $x_0$. An edge between two vertices indicates that these two servers can span $x_0$ together. Then in the proof above, we are actually saying that when $t\ge d^2$, we will have a complete bipartite graph. This constraint is actually not necessary. The essential constraint is only to guarantee a perfect matching in this bipartite graph, i.e., to guarantee that all those servers not containing the singleton $x_0$ could be divided into pairs, with each pair capable of spanning $x_0$. Following this idea, we extend Theorem \ref{first} to a wider range of parameters in the next subsection.
\end{remark}

\subsection{$d^2-d<t<d^2$}

Before the tedious analysis on this range of parameters, we first provide an example illustrating the essence of the proof.

\begin{example}[$d=5,t=23,p=28$]
The corresponding parameters are $\omega=2$, $\mu=14$, $\omega_1=1$, $d_1=5$, $\omega_2=2$ and $d_2=3$. We have $84$ singleton servers: $\overline{A_j}$, $0\le j \le 27$, each appearing three times, where $A_j=\{x_{j+\alpha}:0\le \alpha \le 4\}$. We have $70$ $\Sigma$-servers: $\Sigma B_j$, $0\le j \le 13$, each appearing five times, where $B_j=\{x_{j+5\gamma+14\lambda: 0\le \gamma\le 2, 0\le \lambda \le 1}\}$. The servers not containing the singleton $x_0$ are: $\Sigma B_0$, $\Sigma B_9$ and $\Sigma B_4$, each appearing five times; $\overline{A_0}$, $\overline{A_{27}}$, $\overline{A_{26}}$, $\overline{A_{25}}$ and $\overline{A_{24}}$, each appearing three times. Note that $\Sigma B_0$ cannot be connected to $\overline{A_{24}}$ since $B_0\bigcap A_{24}=\{0,24\}$. Also $\Sigma B_4$ cannot be connected to $\overline{A_{0}}$ since $B_4\bigcap A_0=\{0,4\}$. A perfect matching of the bipartite graph induced by these servers is shown as follows.

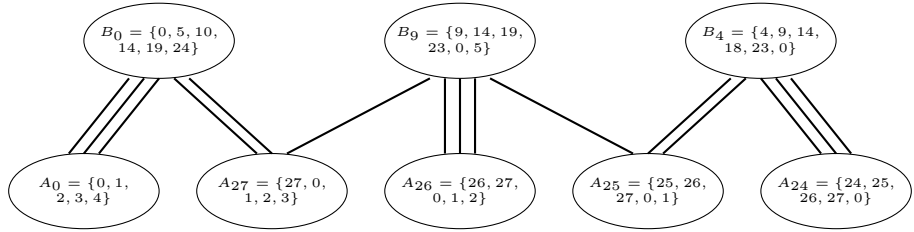
\begin{figure*}[!h]
\centering
\begin{tikzpicture}[scale=0.5]
     \tikzstyle{edge} = [draw,thick,black]
     \draw (-8,2) ellipse [x radius=2, y radius=1];
     \draw (0,2) ellipse [x radius=2, y radius=1];
     \draw (8,2) ellipse [x radius=2, y radius=1];
     \draw (10,-2) ellipse [x radius=2, y radius=1];
     \draw (5,-2) ellipse [x radius=2, y radius=1];
     \draw (0,-2) ellipse [x radius=2, y radius=1];
     \draw (-5,-2) ellipse [x radius=2, y radius=1];
     \draw (-10,-2) ellipse [x radius=2, y radius=1];
     \begin{tiny} \node at (0,2.2)  {$B_9=\{9,14,19,$};\end{tiny}
     \begin{tiny} \node at (0,1.8)  {$23,0,5\}$};\end{tiny}
     \begin{tiny} \node at (-8,2.2)  {$B_0=\{0,5,10,$};\end{tiny}
     \begin{tiny} \node at (-8,1.8)  {$14,19,24\}$};\end{tiny}
     \begin{tiny} \node at (8,2.2)  {$B_4=\{4,9,14,$};\end{tiny}
     \begin{tiny} \node at (8,1.8)  {$18,23,0\}$};\end{tiny}

     \begin{tiny} \node at (-10,-1.8)  {$A_0=\{0,1,$};\end{tiny}
     \begin{tiny} \node at (-10,-2.2)  {$2,3,4\}$};\end{tiny}
     \begin{tiny} \node at (-5,-1.8)  {$A_{27}=\{27,0,$};\end{tiny}
     \begin{tiny} \node at (-5,-2.2)  {$1,2,3\}$};\end{tiny}
     \begin{tiny} \node at (0,-1.8)  {$A_{26}=\{26,27,$};\end{tiny}
     \begin{tiny} \node at (0,-2.2)  {$0,1,2\}$};\end{tiny}
     \begin{tiny} \node at (5,-1.8)  {$A_{25}=\{25,26,$};\end{tiny}
     \begin{tiny} \node at (5,-2.2)  {$27,0,1\}$};\end{tiny}
     \begin{tiny} \node at (10,-1.8)  {$A_{24}=\{24,25,$};\end{tiny}
     \begin{tiny} \node at (10,-2.2)  {$26,27,0\}$};\end{tiny}

     \draw[edge] (-8.8,1)-- (-10.4,-1);
     \draw[edge] (-8.4,1)-- (-10,-1);
     \draw[edge] (-8,1)-- (-9.6,-1);
     \draw[edge] (-7.6,1)--(-5.4,-1);
     \draw[edge] (-7.2,1)-- (-5,-1);
     \draw[edge] (-0.8,1)-- (-4.6,-1);
     \draw[edge] (-0.4,1)-- (-0.4,-1);
     \draw[edge] (0,1)-- (0,-1);
     \draw[edge] (0.4,1)--(0.4,-1);
     \draw[edge] (0.8,1)-- (4.6,-1);
     \draw[edge] (7.2,1)-- (5,-1);
     \draw[edge] (7.6,1)-- (5.4,-1);
     \draw[edge] (8,1)-- (9.6,-1);
     \draw[edge] (8.4,1)--(10,-1);
     \draw[edge] (8.8,1)-- (10.4,-1);

\end{tikzpicture}
\end{figure*}

\end{example}

We shall briefly preview the outline of the proof to come. In the case $d^2-d<t<d^2$, while we stick to the construction in the previous subsection, the bipartite graph induced by those servers not containing the singleton $x_0$ is no longer complete. To deal with this trouble, we shall show that the absent edges are incident to only two kinds of $\Sigma$-servers. To find a perfect matching in the bipartite graph, it suffices to find suitable partners for these two kinds of $\Sigma$-servers first and the rest edges can be chosen arbitrarily.

\begin{lemma} \label{choiceofj}
For $d^2-d < t< d^2$, suppose $|A_{j_1}\bigcap B_{j_2}|>1$ and $0\in A_{j_1}\bigcap B_{j_2}$, then $j_2=0$ or $j_2=\mu\omega_1-d_1(d_2-1)$.
\end{lemma}


\begin{proof}
There exist $\alpha_1,\alpha_2,\beta_1,\beta_2,\gamma_1,\gamma_2,\lambda_1,\lambda_2$ such that $$j_1+\alpha_1+\beta_1\mu\omega_2\equiv j_2+\gamma_1d_1+\lambda_1\mu\omega_1 \equiv0\pmod{p}$$ and $$j_1+\alpha_2+\beta_2\mu\omega_2\equiv j_2+\gamma_2d_1+\lambda_2\mu\omega_1 \pmod{p}.$$
Do subtractions using these two equations above, we can get $$(\alpha_2-\alpha_1)+(\beta_2-\beta_1)\mu\omega_2\equiv (\gamma_2-\gamma_1)d_1+(\lambda_2-\lambda_1)\mu\omega_1 \pmod{p}.$$
Then we have $(\alpha_2-\alpha_1)\equiv (\gamma_2-\gamma_1)d_1 \pmod{\mu}$. When $t>d^2-d$, $\mu=\frac{p}{\omega_1\omega_2}> \frac{d^2}{\omega_1\omega_2}\ge d_1(d_2-1)$. Since $1-d_1\le \alpha_2-\alpha_1 \le d_1-1$ and $1-d_2 \le \gamma_2-\gamma_1 \le d_2-1$, so $(\alpha_2-\alpha_1)\equiv (\gamma_2-\gamma_1)d_1 \pmod{\mu}$ holds if and only if one of the following holds:

$\bullet$ Case I. $\alpha_2-\alpha_1=\gamma_2-\gamma_1=0$. Then by the same analysis as in Theorem \ref{first} we will arrive at a contradiction to $|A_{j_1}\bigcap B_{j_2}|>1$. So this case is impossible.

$\bullet$ Case II. $\gamma_2-\gamma_1=d_2-1$, and consequently $\gamma_2=d_2-1$ and $\gamma_1=0$. Then we have $j_2+\lambda_1\mu\omega_1 \equiv0\pmod{p}$. Since $0\le j_2\le \mu\omega_1-1$ and $0\le\lambda_1\le\omega_2-1$, then we must have $j_2=0$.

$\bullet$ Case III. $\gamma_2-\gamma_1=1-d_2$, and consequently $\gamma_1=d_2-1$ and $\gamma_2=0$. Then we have $j_2+d_1(d_2-1)+\lambda_1\mu\omega_1 \equiv0\pmod{p}$. Since $0\le j_2\le \mu\omega_1-1$, $0\le\lambda_1\le\omega_2-1$ and $d_1(d_2-1)< \mu\le \mu\omega_1$, then we must have $j_2=\mu\omega_1-d_1(d_2-1)$.
\end{proof}

So we only need to focus on two kinds of special $\Sigma$-servers, $\Sigma B_0$ and $\Sigma B_{\mu\omega_1-d_1(d_2-1)}$.

\begin{lemma} \label{firstspecialserver}
$A_{j_1}\bigcap B_0=\{x_0\}$ if and only if $j_1=0$ or $\mu\omega_2-\mu+d_1(d_2-1)+1\le j_1 \le \mu\omega_2-1$.
\end{lemma}

\begin{proof}
$j_1+\alpha_1+\beta_1\mu\omega_2\equiv0 \pmod{p}$ holds if and only if $$j_1=\alpha_1=\beta_1=0$$ or $$\beta_1=\omega_1-1 \text{~and~} j_1=\mu\omega_2-\alpha_1.$$ So the candidate for $j_1$ satisfying $x_0\in A_{j_1}$ is $j_1\in\{0\}\bigcup[\mu\omega_2-d_1+1,\mu\omega_2-1]$.

We should then exclude those $j_1$ such that $|A_{j_1}\bigcap B_0|>1$. Continuing Case II in Lemma \ref{choiceofj}, $\gamma_2-\gamma_1=d_2-1$, then $\alpha_2-\alpha_1=d_1(d_2-1)-\mu$. So $\alpha_1\in[0,d_1-1]\bigcap[\mu-d_1(d_2-1),\mu-d_1(d_2-1)+d_1-1]=[\mu-d_1(d_2-1),d_1-1]$. Then the candidate for $j_1$ such that $|A_{j_1}\bigcap B_0|>1$ is $j_1\in[\mu\omega_2-d_1+1,\mu\omega_2-\mu+d_1(d_2-1)]$. Therefore, by excluding these choices for $j_1$, we finally deduce that $A_{j_1}\bigcap B_0=\{x_0\}$ if and only if $j_1=0$ or $\mu\omega_2-\mu+d_1(d_2-1)+1\le j_1 \le \mu\omega_2-1$.
\end{proof}

\begin{lemma} \label{secondspecialserver}
$A_{j_1}\bigcap B_{\mu\omega_1-d_1(d_2-1)}=\{x_0\}$ if and only if $\mu\omega_2-d_1+1\le j_1 \le \mu\omega_2-d_1d_2+\mu$.
\end{lemma}

\begin{proof}
$j_1+\alpha_1+\beta_1\mu\omega_2\equiv0 \pmod{p}$ holds if and only if $$j_1=\alpha_1=\beta_1=0$$ or $$\beta_1=\omega_1-1 \text{~and~} j_1=\mu\omega_2-\alpha_1.$$ So the candidate for $j_1$ satisfying $x_0\in A_{j_1}$ is $j_1\in\{0\}\bigcup[\mu\omega_2-d_1+1,\mu\omega_2-1]$.

We should then exclude those $j_1$ such that $|A_{j_1}\bigcap B_{\mu\omega_1-d_1(d_2-1)}|>1$. Continuing Case III in Lemma \ref{choiceofj}, $\gamma_2-\gamma_1=1-d_2$, then $\alpha_2-\alpha_1=\mu-d_1(d_2-1)$. So $\alpha_1\in[0,d_1-1]\bigcap[d_1(d_2-1)-\mu,d_1(d_2-1)-\mu+d_1-1]=[0,d_1(d_2-1)-\mu+d_1-1]$. Then the candidate for $j_1$ such that $|A_{j_1}\bigcap B_0|>1$ is $j_1\in[\mu\omega_2-d_1(d_2-1)+\mu-d_1+1,\mu\omega_2-1]\bigcup\{0\}$. Therefore, by excluding these choices for $j_1$, we finally deduce that $A_{j_1}\bigcap B_{\mu\omega_1-d_1(d_2-1)}=\{x_0\}$ if and only if $\mu\omega_2-d_1+1\le j_1 \le \mu\omega_2-d_1d_2+\mu$.
\end{proof}

\begin{lemma} \label{inequality}
$d_2(\mu-d_1(d_2-1))\ge d_1.$
\end{lemma}

\begin{proof}
This is equivalent to $\frac{d+1}{\omega_2}(\frac{p}{\omega_1\omega_2}-\frac{d}{\omega_1}(\frac{d+1}{\omega_2}-1))\ge\frac{d}{\omega_1}$. Reorganizing this inequality we get $d\omega_2(d+1-\omega_2)\ge(d+1)(d^2+d-p)$. Since $\omega_2|d+1$ and $d^2<p<d^2+d$, so $p$ is not a multiple of $d+1$ and therefore $\omega_2\neq d+1$. Thus we have $1\le \omega_2 \le \frac{d+1}{2}$, then the left-hand-side is at least $d^2$. Since $d^2<p$, we have $d^2+d-p\le d-1$ and the right-hand-side is at most $d^2-1$. Therefore the inequality holds.
\end{proof}

Combining these lemmas above, we can finally extend Theorem \ref{first} to the range of parameters $d^2-d<t<d^2$.

\begin{theorem}
For $d^2-d< t< d^2$, there exist $k$-PIR array codes with $m=\mu(2d+1)$ servers such that the rate $\frac{k}{m}$ equals $g(s,t)=1-\frac{d^2+d}{p(2d+1)}$.
\end{theorem}

\begin{proof}
It suffices to find a perfect matching in the bipartite graph induced by those servers not containing the singleton $x_0$. For those $d_1$ servers named $\Sigma B_0$, from Lemma \ref{firstspecialserver} we know that each $\Sigma B_0$ is connected to the singleton server $\overline{A_j}$ with $j\in\mathcal{S}\triangleq[\mu\omega_2-\mu+d_1(d_2-1)+1, \mu\omega_2-1]\bigcup\{0\}$. Since each $A_j$ appears $d_2$ times so there are totally $d_2(\mu-d_1(d_2-1))$ such servers.

Similarly, for those $d_1$ servers named $\Sigma B_{\mu\omega_1-d_1(d_2-1)}$, from Lemma \ref{secondspecialserver} we know that each $\Sigma B_{\mu\omega_1-d_1(d_2-1)}$ is connected to the singleton server $\overline{A_j}$ with $j\in\mathcal{T}\triangleq[\mu\omega_2-d_1+1,\mu\omega_2-d_1d_2+\mu]$. Since each $A_j$ appears $d_2$ times so there are totally $d_2(\mu-d_1(d_2-1))$ such servers.

For any $\Sigma$-server other than $\Sigma B_0$ and $\Sigma B_{\mu\omega_1-d_1(d_2-1)}$, Lemma \ref{choiceofj} tells us that it is connected to all the singleton servers not containing the singleton $x_0$. Therefore, in order to find a perfect matching, we only need to find the edges incident with the servers $\Sigma B_0$ and $\Sigma B_{\mu\omega_1-d_1(d_2-1)}$, and the rest edges can be chosen arbitrarily. By Lemma \ref{inequality}, $d_2|\mathcal{S}|=d_2|\mathcal{T}|=d_2(\mu-d_1(d_2-1))\ge d_1$. Moreover, $d_2|\mathcal{S}\bigcup\mathcal{T}|=d_1d_2\ge2d_1$, where $d_2\ge2$ follows from the fact that $\omega_2\neq d+1$ shown in Lemma \ref{inequality}. Therefore finding the partners for those servers named $\Sigma B_0$ and $\Sigma B_{\mu\omega_1-d_1(d_2-1)}$ could be done and the result follows.
\end{proof}

\section{$s>2$: a new upper bound of $g(s,t)$} \label{newupperbound}

In this section we derive a new upper bound of $g(s,t)$ for $s>2$ (equivalently, $d>t$ and $p=d+t>2t$), improving the original upper bound shown in Theorem \ref{upperbound}.

For any given PIR array code, we first divide the servers into the following four parts. The first part contains all the $l$ singleton servers, i.e., servers whose cells are all singleton entries. The second part contains all the $r$ servers, where each server has $t-1$ singleton entries and the remaining entry is a summation of $\eta$ out of the left $p-t+1$ items, $2\le \eta\le t+1$. The third part contains all the $u$ servers, where each server has $t-1$ singleton entries and the remaining entry is a summation of $\lambda$ out of the left $p-t+1$ items, $t+1<\lambda\le p-t+1$. Finally the fourth part contains all the $w$ servers, where each server has at most $t-2$ singleton entries. Clearly $l+r+u+w=m$.

\begin{theorem} \label{mainupperbound}
For any integer $t\ge2$ and any positive integer $d>t$, we have
$$g(1+\frac{d}{t},t)\le\frac{d^2+2t^2+3td+2t}{2(t+d)(d+t+1)}.$$
\end{theorem}

\begin{proof}
Suppose we have a $[t\times m,p]$ $k$-PIR array code with parameters satisfying the condition of the theorem. For each $i\in\{1,2,\dots,p\}$, let $S^i_{1},S^i_{2},\dots,S^i_{k_i}$ be disjoint subsets of servers such that each subset of servers could span the item $x_i$. The number $k_i$ is chosen to be as large as possible. To derive an upper bound on $k/m$, it suffices to show that
$$\sum_{i=1}^p k_i\le \frac{d^2+2t^2+3td+2t}{2(d+t+1)}m.$$

Among each of the four parts, without loss of generality, any server containing a singleton entry $x_i$ can be chosen as one of the subsets $S^i_{1},S^i_{2},\dots,S^i_{k_i}$. Assume the numbers of such servers among the four parts are $l_i$, $r_i$, $u_i$ and $w_i$. Let $f_i$ be the number of subsets in the list $S^i_{1},S^i_{2},\dots,S^i_{k_i}$ consisting of exactly one singleton server and one non-singleton server. Let $g_i$ be the number of subsets in the list $S^i_{1},S^i_{2},\dots,S^i_{k_i}$ containing at least two singleton servers and exactly one non-singleton server. For any remaining subset in the list $S^i_{1},S^i_{2},\dots,S^i_{k_i}$ other than those listed above, it must contain at least two non-singleton servers. So we have the following inequality:
$$k_i\le l_i+r_i+u_i+w_i+f_i+g_i+\frac{r-r_i+u-u_i+w-w_i-f_i-g_i}{2}.$$

Below we estimate $\sum k_i$ in two ways. First, by counting the singleton servers we have $f_i+2g_i\le l-l_i$. So we have
\begin{align*}
k_i&\le l_i+r_i+u_i+w_i+f_i+g_i+\frac{r-r_i+u-u_i+w-w_i-f_i-g_i}{2}\\
&= l_i+\frac{r+r_i+u+u_i+w+w_i}{2}+\frac{f_i}{2}+\frac{g_i}{2}\\
&\le l_i+\frac{r+r_i+u+u_i+w+w_i}{2}+\frac{f_i}{2}+g_i\\
&\le l_i+\frac{r+r_i+u+u_i+w+w_i}{2}+\frac{l-l_i}{2}\\
&=\frac{l+l_i+r+r_i+u+u_i+w+w_i}{2}.
\end{align*}

By counting the number of singleton cells in each of the four parts of servers, we have $\sum l_i=lt$, $\sum r_i=r(t-1)$, $\sum u_i=u(t-1)$ and $\sum w_i\le w(t-2)$. These lead to
\begin{align}
\sum k_i &\le \frac{p(l+r+u+w)}{2}+\frac{lt+r(t-1)+u(t-1)+w(t-2)}{2} \notag \\
&= l\frac{p+t}{2}+r\frac{p+t-1}{2}+u\frac{p+t-1}{2}+w\frac{p+t-2}{2}.
\end{align}

The second estimation is to analyze $f_i$, the number of subsets in the list $S^i_{1},S^i_{2},\dots,S^i_{k_i}$ containing exactly one singleton server and one non-singleton server. Notice that the non-singleton server cannot be from the third part. This is because such a server from the third part does not have the singleton entry $x_i$ and its unique non-singleton cell should be of the form $x_i+\sum_{j=1}^{\lambda-1} y_j$, where $\lambda-1>t$. Any singleton server without the singleton entry $x_i$ could only provide the values of $t$ items. So they two cannot cooperate on spanning $x_i$. Therefore, trivially we have $f_i\le r-r_i+w-w_i$ and $\sum f_i\le pr-r(t-1)+pw-\sum w_i$.

However, this is still not enough. The following observation will be the key to this theorem. For any non-singleton server from the second part, its unique non-singleton entry is a summation of at most $t+1$ items. So its contribution to counting $\sum f_i$ is at most $t+1$. Therefore, instead of using $\sum f_i\le pr-r(t-1)+pw-\sum w_i$, a better estimation is  $\sum f_i\le r(t+1)+pw-\sum w_i$. Then we have

\begin{align*}
k_i&\le l_i+r_i+u_i+w_i+f_i+g_i+\frac{r-r_i+u-u_i+w-w_i-f_i-g_i}{2}\\
&= l_i+\frac{r+r_i+u+u_i+w+w_i}{2}+\frac{f_i}{4}+\frac{g_i}{2}+\frac{f_i}{4}\\
&\le l_i+\frac{r+r_i+u+u_i+w+w_i}{2}+\frac{l-l_i}{4}+\frac{f_i}{4}\\
&=\frac{l+3l_i+2r+2r_i+2u+2u_i+2w+2w_i}{4}+\frac{f_i}{4},
\end{align*}
and thus

\begin{align}
\sum k_i &\le \sum\frac{l+3l_i+2r+2r_i+2u+2u_i+2w+2w_i}{4}+\frac{\sum f_i}{4}  \notag\\
&=\frac{lp+3lt+2rp+2r(t-1)+2up+2u(t-1)+2wp+2\sum w_i}{4}+\frac{r(t+1)+wp-\sum w_i}{4} \notag\\
&= \frac{lp+3lt+2rp+2r(t-1)+2up+2u(t-1)+2wp+r(t+1)+wp}{4}+\frac{\sum w_i}{4} \notag \\
&\le l\frac{p+3t}{4}+r\frac{2p+3t-1}{4}+u\frac{p+t-1}{2}+w\frac{3p+t-2}{4}.
\end{align}

Now we have estimated $\sum k_i$ in two ways, the formula $(1)$ and the formula $(2)$. Denote $$F(l,r,u,w)=l\frac{p+t}{2}+r\frac{p+t-1}{2}+u\frac{p+t-1}{2}+w\frac{p+t-2}{2},$$ and denote $$G(l,r,u,w)=l\frac{p+3t}{4}+r\frac{2p+3t-1}{4}+u\frac{p+t-1}{2}+w\frac{3p+t-2}{4},$$ then $\sum k_i \le \min \{F,G\}$. To find an upper bound of $\sum k_i$, it suffices to determine the maximum value of $\min \{F,G\}$. Suppose that this maximum occurs at $(\widetilde{l},\widetilde{r},\widetilde{u},\widetilde{w})$. It is easy to check that $F(\widetilde{l},\widetilde{r},\widetilde{u},\widetilde{w}) \le F(\widetilde{l},\widetilde{r}+\widetilde{u},0,\widetilde{w})$ and $G(\widetilde{l},\widetilde{r},\widetilde{u},\widetilde{w}) \le G(\widetilde{l},\widetilde{r}+\widetilde{u},0,\widetilde{w})$. Therefore we have $\widetilde{u}=0$. Now the problem reduces to
\begin{align*}
\max ~~~~& \min\Bigg\{l\frac{p+t}{2}+r\frac{p+t-1}{2}+w\frac{p+t-2}{2},~l\frac{p+3t}{4}+r\frac{2p+3t-1}{4}+w\frac{3p+t-2}{4} \Bigg\}, \\
s.t. ~~~~& l+r+w=m,~l,r,w\in\mathbb{N} \\
\end{align*}

To solve the program above, first let $w$ be fixed. Then it is routine to deduce that when $l=m\frac{t+1}{p+1}+w\frac{p-2t+1}{p+1}$, the corresponding objective function is then a function of $w$ of the form $m\frac{p^2+tp+2t}{2(p+1)}-w\frac{t}{(p+1)}$. Then to maximize this function we will have $w=0$. To sum up, when $w=0$, $l=m\frac{t+1}{p+1}$ and $r=m\frac{p-t}{p+1}$, the optimal value of the program is then $\frac{d^2+2t^2+3td+2t}{2(d+t+1)}m$ and the theorem follows.
\end{proof}

\smallskip
Finally, it is straightforward to check that our new upper bound is better than Theorem \ref{upperbound} when $s>2$. That is, $\frac{d^2+2t^2+3td+2t}{2(t+d)(d+t+1)}<\frac{(2d+1)t+d^2}{(t+d)(2d+1)}$ when $d>t$.

\section{$s>2$: analyzing the construction by Blackburn and Etzion} \label{constru}

For the case $s>2$, Blackburn and Etzion propose a new PIR array code in \cite[Section 4]{blackburn}. We briefly review their construction (hereafter we call it the B-E Construction) for $s$ being an integer. The case when $s$ is not an integer can be managed similarly.

For a server with $t-1$ singleton cells and the other one containing a summation of $j$ out of the left $st-t+1$ items, we shall call it a server of type $j$, $1\le j \le st-t+1$. A type $1$ server is then just a singleton server. For $1\le r \le s$, let $T_r$ denote the whole set of servers of type $(r-1)t+1$ containing all the possible combinations of singleton cells and the summation cell. That is, $|T_1|={st\choose t}$ and $|T_r|={st\choose t-1}{{st-t+1}\choose{(r-1)t+1}}$ for $2\le r \le s$. The B-E construction consists of servers $T_r$, $1\le r \le s$, with each $T_r$ appearing $\eta_r$ times. For any given item $x_i$, we pair those servers not containing the singleton cell $x_i$ by constructing $s-1$ bipartite graphs. The choices for $\eta_r$ are to guarantee that, for any given item $x_i$, each of the $s-1$ bipartite graphs has a perfect matching, i.e., all those servers not containing the singleton $x_i$ can be divided into pairs, each pair capable of spanning $x_i$.

For any given item $x_i$, the $s-1$ bipartite graphs are as follows. The bipartite graph $G_r$, $1\le r \le s-1$, has two sides. The first side represents all the servers in $T_r$ (appearing $\eta_r$ times), in which $x_i$ is neither a singleton nor appears in the summation part. The second side represents all the servers in $T_{r+1}$ (appearing $\eta_{r+1}$ times), in which $x_i$ appears in the summation part. An edge is connected between $v$ from the first side and $u$ from the second side, if and only if all the items appeared in $v$ ($t-1$ singletons and $(r-1)t+1$ items in the summation part) are exactly the $rt$ items in the summation part of $u$, excluding $x_i$. It can be easily calculated that, to guarantee a perfect matching in $G_r$, we must have the ratio $\eta_1:\eta_2={{p-t-1}\choose{t-1}}:1$ and $\eta_r:\eta_{r+1}={{p-rt-1}\choose{t-1}}:{{rt}\choose{t-1}}$ for $2\le r \le s-1$.

\begin{example}[$s=4$] \label{example}
The servers in $T_1,T_2,T_3,T_4$ appear $\eta_1,\eta_2,\eta_3,\eta_4$ times respectively. The following ratios are required. $\eta_1:\eta_2={{3t-1}\choose{t-1}}:1$, $\eta_2:\eta_3=(t+1):2t$ and $\eta_3:\eta_4=1:{{3t}\choose{t-1}}$. So we may select $\eta_1=(t+1){{3t-1}\choose{t-1}}$, $\eta_2=t+1$, $\eta_3=2t$ and $\eta_4=2t{{3t}\choose{t-1}}$.
\end{example}

Compared to those existing PIR array codes in Theorem \ref{severalconstr}, one advantage of the B-E Construction is that it is a unified construction, suitable to all parameters.

In this section we shall deeply analyze the B-E Construction in several aspects. The rest of this section is divided into three subsections. In the first subsection, based on a slight modification, we present a new construction and compare it with the original B-E Construction. In the second subsection we show that both constructions produce PIR array codes with better rate than all the other existing PIR array codes in Theorem \ref{severalconstr}. Finally in the third subsection, we add some remarks regarding the possible optimality of the B-E Construction.

\subsection{A new construction}

A minor problem of the B-E Construction is that, the number of servers $m$ could be very large since the choices for the integers $\{\eta_r:1\le r \le s\}$ should abide by the desired ratio. Based on a slight modification, we present a new construction of PIR array codes, with a much smaller number of servers compared to the B-E Construction, with a slight sacrifice in the rate.\\

\noindent\fbox{%
  \parbox{\textwidth}{%
Construction (given $t$, $s>2$ and $p=ts$):\\

1. Take all those ${p\choose t}$ singleton servers, each appearing $\delta$ times, where $\delta={{p-t-1}\choose{t-1}}$.\\

2. For a server with $t-1$ singleton cells and the other one containing a summation of $j$ out of the left $p-t+1$ items, we shall call it a server of type $j$, $2\le j \le p-t+1$. Take all those servers of types $t+1, t+2,\dots,p-t+1$, each appearing exactly once.
  }%
}

\medskip

It is easy to see that the array code consisting of these servers is indeed symmetric for all the items $\{x_1,x_2,\dots,x_p\}$. To span any symbol, say $x_i$, all those servers containing a singleton cell $x_i$ will surely do, and we hope that those servers not containing a singleton cell $x_i$ can be divided into pairs and each pair could together span $x_i$. This is shown in the following theorem.

\begin{theorem} \label{main}
The construction above gives an array code with $m={p\choose t}{{p-t-1}\choose{t-1}}+{p\choose {t-1}} \sum_{t+1\le j \le p-t+1} {{p-t+1}\choose j}$ and $k=\frac{p+t}{2p}{p\choose t}{{p-t-1}\choose{t-1}}+\frac{p+t-1}{2p}{p\choose {t-1}} \sum_{t+1\le j \le p-t+1} {{p-t+1}\choose j}$.
\end{theorem}

\begin{proof}
There are totally ${p\choose t}\times\delta$ singleton servers and totally ${p\choose {t-1}}{{p-t+1}\choose j}$ servers of type $j$, altogether
$$m={p\choose t}{{p-t-1}\choose{t-1}}+{p\choose {t-1}} \sum_{t+1\le j \le p-t+1} {{p-t+1}\choose j}.$$

Among the singleton servers, the number of singleton cells containing the singleton $x_i$ is exactly $\frac{t}{p}{p\choose t}{{p-t-1}\choose{t-1}}$. Among the non-singleton servers (each with $t-1$ singleton cells), the number of singleton cells containing the singleton $x_i$ is exactly $\frac{t-1}{p}{p\choose {t-1}} \sum_{t+1\le j \le p-t+1} {{p-t+1}\choose j}$.

For each singleton server without the singleton cell $x_i$, suppose it stores $\{y_1,y_2,\dots,y_t\}$. Then we may pair it with a server of type $t+1$ whose unique non-singleton cell contains $x_i+\sum_{j=1}^t y_j$. These two servers could together span $x_i$. The singleton server storing $\{y_1,y_2,\dots,y_t\}$ appears $\delta={{p-t-1}\choose{t-1}}$ times. Meanwhile we also have exactly ${{p-t-1}\choose{t-1}}$ servers of type $t+1$ whose unique non-singleton cell contains $x_i+\sum_{j=1}^t y_j$. So clearly we may divide these two families of servers in pairs.

Next we analyze a server of type $j$ satisfying: 1) it does not contain the singleton $x_i$ and 2) $x_i$ does not appear in the summation in its non-singleton cell, for some $t+1\le j < p-t+1$. We may suppose this server contains $z_1,z_2,\dots,z_{t-1}$ and $\omega_1+\omega_2+\dots+w_j$. Then we can pair it with a server of type $j+1$ containing $z_1,z_2,\dots,z_{t-1}$ and $x_i+\omega_1+\omega_2+\dots+w_j$. Clearly they two will together derive the value of $x_i$.

In this way, all the servers not containing a singleton $x_i$ are divided into pairs and each pair could together span $x_i$. So we have
\begin{align*}
k&=\frac{t}{p}{p\choose t}{{p-t-1}\choose{t-1}}+\frac{t-1}{p}{p\choose {t-1}} \sum_{t+1\le j \le p-t+1} {{p-t+1}\choose j}+\frac{m-\frac{t}{p}{p\choose t}{{p-t-1}\choose{t-1}}-\frac{t-1}{p}{p\choose {t-1}} \sum_{t+1\le j \le p-t+1} {{p-t+1}\choose j}}{2}\\
&=\frac{p+t}{2p}{p\choose t}{{p-t-1}\choose{t-1}}+\frac{p+t-1}{2p}{p\choose {t-1}} \sum_{t+1\le j \le p-t+1} {{p-t+1}\choose j}.
\end{align*}
\end{proof}

\begin{example}[$s=4$]
We have all combinations of ${{4t}\choose t}$ singleton servers, each appearing ${{3t-1}\choose{t-1}}$ times. Then we have all the servers of type $t+1,t+2,\dots,3t+1$, totally ${{4t}\choose{t-1}}\sum_{t+1\le j \le 3t+1} {{3t+1}\choose j}$ servers. It is routine to check that the number of servers is much smaller than that of Example \ref{example}.
\end{example}

\subsection{Comparing the rate of different PIR array codes} \label{rate}

To calculate the exact rate of the B-E Construction or our new construction is tedious and in some sense, not necessary. First note that in the B-E Construction and our modified construction, all those servers not containing $x_i$ could be divided into pairs so that each pair is capable of spanning $x_i$. Suppose we have $\alpha$ singleton servers and $\beta$ non-singleton servers. Among the $\alpha$ singleton servers there are totally $t\alpha$ singleton cells. Since the code is obviously symmetric for all the items, then $t\alpha/p$ servers contain a singleton cell $x_i$. Similarly among the $\beta$ non-singleton servers (each has $t-1$ singleton cells and a summation cell) there are $(t-1)\beta/p$ servers containing a singleton cell $x_i$. So we have $m=\alpha+\beta$ and $k=t\alpha/p+(t-1)\beta/p+\frac{m-t\alpha/p-(t-1)\beta/p}{2}$. Thus the rate $k/m=\frac{t+p}{2p}\frac{\alpha}{\alpha+\beta}+\frac{t+p-1}{2p}\frac{\beta}{\alpha+\beta}$ is a weighted average of $\frac{t+p}{2p}$ and $\frac{t+p-1}{2p}$ and is strictly larger than $\frac{t+p-1}{2p}=\frac{ts+t-1}{2ts}$.

At this point an easy observation is that a larger ratio $\frac{\alpha}{\alpha+\beta}$ implies a larger rate. It is tedious but straight forward to check that the B-E Construction does have a larger ratio $\frac{\alpha}{\alpha+\beta}$ compared to our modified construction. So the rate of the B-E Construction is strictly larger than ours. Next we shall show that our construction produces codes with better rate than the other existing PIR array codes in Theorem \ref{severalconstr}.

Since our construction produces codes with rate strictly larger than $\frac{t+p-1}{2p}=\frac{ts+t-1}{2ts}$, we shall first use this value to compare with the other existing PIR array codes and successfully show that $\frac{ts+t-1}{2ts}$ is larger than the rate of most existing codes. However there exist some sporadic cases when this comparison is not enough and then we have to proceed with a more detailed comparison.\\

$\bullet$ Comparison with Theorem \ref{severalconstr}, the first case (Construction 7 and Theorem 10 in \cite{blackburn}):\\

Let $3\le r \le t$ and $s=r-\frac{(r-2)+1}{t}$. In this case $p=ts=tr-t^2+2r-1$. It is easy to check that $\frac{ts+t-1}{2ts}>\frac{1}{2}+\frac{t-r+1}{2(rt-(r-2)r-1)}$ holds when $r\ge 3$.\\

$\bullet$ Comparison with Theorem \ref{severalconstr}, the second case (Construction 8 and Theorem 11 in \cite{blackburn}):\\

Let $s=r+\frac{d}{t}$ and $p=rt+d$, where $r\ge 2$ is an integer, $t\ge r$, $1\le d \le t-1$. The process to show that $\frac{ts+t-1}{2ts}=\frac{rt+d+t-1}{2(rt+d)}$ is larger than $1-\frac{(rt+d-t+r)(rt+d-t)}{(rt+d)(2rt+2d-2t+r)}$ can be reduced to proving $(r-2)(rt+d-t)>r$. This always holds when $r\ge 3$.

The remaining case is when $r=2$, $p=2t+d$ and we have to follow a detailed analysis. In our modified construction, the number of singleton servers is $A={{2t+d}\choose t}{{t+d-1}\choose{t-1}}$. The number of the other servers is $B={{2t+d}\choose t-1}\sum_{i=0}^d {{t+d+1}\choose i}$. So we have the rate
$$k/m=\frac{A\frac{3t+d}{4t+2d}+B\frac{3t+d-1}{4t+2d}}{A+B}=\frac{3t+d}{4t+2d}-\frac{1}{(4t+2d)(\frac{A}{B}+1)}.$$

A lower bound of $\frac{A}{B}$ can be derived as follows, where the third inequality is due to $\sum_{i=0}^d {{t+d+1}\choose i} \le \frac{d(t+d+1)!}{d!(t+1)!}$ and the fourth inequality is due to $d\le t-1$:

$$\frac{A}{B}=\frac{{{2t+d}\choose t}{{t+d-1}\choose{t-1}}}{{{2t+d}\choose t-1}\sum_{i=0}^d {{t+d+1}\choose i}}=\frac{(t+d+1)(t+d-1)!}{t(t-1)!d! \sum_{i=0}^d {{t+d+1}\choose i}}\ge \frac{t+1}{d(t+d)} >\frac{1}{t+d}.$$

Thus the rate $\frac{k}{m}\ge \frac{3t+d}{4t+2d}-\frac{1}{(4t+2d)(\frac{1}{t+d}+1)}$ and it is routine to check that the right-hand side is larger than $1-\frac{(t+d+2)(t+d)}{(2t+d)(2t+2d+2)}$.\\

$\bullet$ Comparison with Theorem \ref{severalconstr}, the third case (Construction 9 and Theorem 12 in \cite{blackburn}):\\

Let $s>2$ be an integer and $t\ge s$. The inequality $\frac{ts+t-1}{2ts}>\frac{ts+t+1}{s(2t+1)}$ can be equivalently reduced to $ts>3t+1$, which naturally holds when $s\ge 4$. The remaining case $s=3$ is analyzed as follows.

When $s=3$, then $p=3t$. In our construction $m={{3t}\choose t}{{2t-1}\choose{t-1}}+{{3t}\choose{t-1}}2^{2t}$ and $k=\frac{2}{3}{{3t}\choose t}{{2t-1}\choose{t-1}}+\frac{4t-1}{6t}{{3t}\choose{t-1}}2^{2t}$. We now prove that our rate $k/m$ is larger than $\frac{4t+1}{6t+3}$ using the following deductions:
\begin{align*}
 \frac{\frac{2}{3}{{3t}\choose t}{{2t-1}\choose{t-1}}+\frac{4t-1}{6t}{{3t}\choose{t-1}}2^{2t}}{{{3t}\choose t}{{2t-1}\choose{t-1}}+{{3t}\choose{t-1}}2^{2t}} &> \frac{4t+1}{6t+3}\\
\Longleftrightarrow \frac{2t+1}{t}{{2t-1}\choose{t-1}} \cdot \frac{1}{6t+3} &> 2^{2t}(\frac{4t+1}{6t+3}-\frac{4t_1}{6t})\\
\Longleftrightarrow (4t+2){{2t-1}\choose{t-1}} &> 2^{2t}
\end{align*}
where the last inequality can be easily checked by induction: when $t=1$, the inequality corresponds to $6>4$; the inductive step follows from $\frac{(4t+6){{2t+1}\choose{t}}}{(4t+2){{2t-1}\choose{t-1}}}=\frac{4t+6}{t+1}>4$.\\

$\bullet$ Comparison with Theorem \ref{severalconstr}, the fourth case (Construction 10 and Theorem 13 in \cite{blackburn}):\\

Let $s>2$ be an integer. Let $(s-1)t=lb$ and $t\ge l+b$, where $l$ and $b$ are positive integers. Obviously $l$ should be larger than 1. Thus $\frac{ts+t-1}{2ts}=\frac{s+1}{2s}-\frac{1}{2st}>\frac{s+1}{2s}-\frac{l}{2st}$ holds.

Summing up the above, we have shown that the PIR array code produced by our construction has better rates than all the existing ones in Theorem \ref{severalconstr}.

\subsection{Does the B-E Construction have optimal rate?}

Finally we add some discussions on the B-E Construction regarding its potential optimality. The following analysis is based on intuitive ideas rather than strict proofs. To prove or disprove the optimality of the B-E Construction will be of great interest.

For any PIR array code with optimal rate, first note that we may assume that all servers of the same type appear the same number of times, just as in the B-E Construction. This is because if it is not the case, then we may choose any permutation $\pi\in S_p$ and let it operate on the code, exchanging the names of the items. Taking the union of all such $p!$
codes will result in an optimal code, in which all servers of the same type appear the same number of times.

So we may assume that all those ${p\choose t}$ singleton servers appear a certain number of times. For each of those singleton servers not containing a given item $x_i$, we shall find its partner to cooperate on spanning $x_i$. We turn to servers of type $j$ for help and the candidate for the value $j$ should be $2\le j \le t+1$. Then what should be the proper choice for $j$? The set of servers of type $j$ can be divided into three subsets: $A_j$ servers containing a singleton $x_i$, $B_j$ servers containing $x_i$ in its summation part and the rest $C_j$ servers in which $x_i$ neither appears as a singleton nor appears in its summation part. It can be easily calculated that $\frac{B_j}{C_j}=\frac{1}{\frac{p-t+1}{j}-1}$, which increases as $j$ increases. The $B_j$ servers are the partners we wish to find for those singleton servers not containing $x_i$ and the $C_j$ servers accompanied will become new troubles. So intuitively we wish to maximize the ratio $\frac{B_j}{C_j}$ and thus $j=t+1$. Following the same analysis step-by-step, we choose the servers of type $t+1,2t+1,3t+1\dots$, which is exactly the B-E Construction.

Moreover, it seems that bringing in servers with less than $t-1$ singleton cells does no good. Suppose we have $\alpha_i$ servers with $t-i$ singleton cells, $0\le i\le t$. Then following a similar analysis as in Subsection \ref{rate}, the rate will be a weighted average of $\{\frac{t+p-i}{2p}:0\le i \le t\}$, where $\alpha_i$'s are the corresponding weights. So the existence of servers with less than $t-1$ singleton cells is very likely to decrease the rate.

To sum up, we believe that all these intuitive analyses above are positive evidences to the following conjecture:
\begin{conjecture}
For $s>2$, the PIR array codes produced by the B-E Construction have optimal rate.
\end{conjecture}

\section{Conclusion} \label{conclusion}

In this paper we consider the problem of constructing optimal PIR array codes, following the work of \cite{fazeli} and \cite{blackburn}. For the case $1<s\le2$, we determine the minimum number of servers admitting a PIR array code with optimal rate for a certain range of parameters, i.e. $t> d^2-d$. We believe a similar result may be found for the remaining cases by a different approach. For the case $s>2$, we derive a new upper bound on the rate and we analyze the construction by Blackburn and Etzion in several aspects. Especially, to prove or to disprove the optimality of the B-E construction for $s>2$ will be of great interest.

\bibliographystyle{siam}
\bibliography{ref}

\end{document}